\newif\ifdraft \draftfalse
\ifdraft
\documentclass[fleqn]{article}
\usepackage{mystyle,amsmath,mymath,proof,latexsym,amssymb,stmaryrd}
\let\institute=\date
\def\email{email: }
\def\and{\ and\ }
\def\inst#1{}

\newtheorem{theorem}{Theorem}[section]
\newtheorem{proposition}[Th]{Proposition}
\newtheorem{corollary}[Th]{Corollary}
\newtheorem{lemma}[Th]{Lemma}
\newtheorem{definition}[Th]{Definition}
\newenvironment{proof}{{\em Proof.}}{}
\def\qed{$\Box$}
\usepackage[T1]{fontenc}

\usepackage{color}
\usepackage{graphicx}
\else

\documentclass[submission,copyright,creativecommons]{eptcs}
\providecommand{\event}{LTT 2026} % Name of the event you are submitting to
\usepackage{iftex}
\usepackage{mystyle,mymath,proof,latexsym,amssymb,stmaryrd}
\usepackage{amsmath,color}

\newtheorem{lemma}[Th]{Lemma}

\newenvironment{proof}{\noindent{\em Proof.}}{\qed}
\newcommand{\qed}{\hfill{$\square$}}
\fi

\def\Tr{{\rm Tr}}
\def\Lfp{{\rm lfp}}
\def\Eval#1{{[\![ {#1} ]\!]}}
\def\Bar{\overline}
\def\text{\hbox}
\def\LKID{{\rm LKID}}
\def\LKIDomega{{\rm LKID}^\omega}
\def\CLKIDomega{{\rm CLKID}^\omega}

\newcommand{\FOLID}{$\mathrm{FOL}_\mathrm{ID}$}
\newcommand{\LKIDo}{LKID$^\omega$}
\newcommand{\CLKIDo}{CLKID$^\omega$}

\newcommand{\cD}{\mathcal{D}}
\newcommand{\cP}{\mathcal{P}}

\newcommand{\bu}{\Vec{u}}
\newcommand{\bt}{\Vec{t}}
\newcommand{\by}{\Vec{y}}
\newcommand{\bx}{\Vec{x}}
\newcommand{\bz}{\Vec{z}}
\newcommand{\FV}{\mathit{FV}}

\newcommand{\Nat}{\mathbb{N}}

\newcommand{\arity}{\mathrm{arity}}
\newcommand{\semb}[1]{[\![ #1 ]\!]}

\newcommand{\cN}{\mathcal{N}}

\def\Code#1{{\ulcorner {#1}\urcorner}}
\def\PA{{\bf PA}}

\def\LKIDomega{{\rm LKID}^\omega}

\ifdraft
\begin{document}
\sloppy
\hbadness=10000
\input{title_draft}
\else

\title{Truth Predicate of Inductive Definitions and
Logical Complexity of Infinite-Descent Proofs} %TODO Please add

\author{Sohei Ito
\institute{Nagasaki University\\ Nagasaki, Japan}
\email{s-ito@nagasaki-u.ac.jp}
\and
Makoto Tatsuta
\institute{National Institute of Informatics / Sokendai\\
Tokyo, Japan}
\email{tatsuta@nii.ac.jp}
}

\newcommand{\titlerunning}{Truth Predicate \& Infinite Descent}
\newcommand{\authorrunning}{S. Ito \& M.Tatsuta}

\hypersetup{
  bookmarksnumbered,
  pdftitle    = {\titlerunning},
  pdfauthor   = {\authorrunning},
  pdfsubject  = {EPTCS},               % Consider adding a more appropriate subject or description
  % pdfkeywords = {keyword1, keyword2} % Uncomment and enter keywords specific to your paper
}
\begin{document}
\maketitle
\fi

\begin{abstract}

Formal reasoning about inductively defined relations and structures is
widely recognized not only for its mathematical interest but also for
its importance in computer science, and has applications in verifying
properties of programs and algorithms. Recently, several proof systems
of inductively defined predicates based on sequent calculus including
the cyclic proof system CLKID-omega and
the infinite-descent proof system LKID-omega
have attracted much attention. 
Although the relation among their provabilities has been clarified so far, the
logical complexity of these systems has not been much studied. 
The infinite-descent proof system LKID-omega is an infinite proof system 
for inductive definitions and allows infinite paths in proof figures.
It serves as a basis for the cyclic proof system.
This paper
shows that 
the logical complexity of the provability in LKID-omega is (Pi-1-1)-complete.
To show this, first it is shown that
the validity for inductive definitions in standard models is equivalent to
the validity for inductive definitions in standard term models.
Next, using this equivalence, this paper extends the truth predicate of omega-languages, as given in Girard's textbook, to inductive definitions
by employing arithmetical coding of inductive definitions.
This shows that the validity of inductive definitions in standard models is a (Pi-1-1) relation.
Then, using the completeness of LKID-omega for standard models,
it is shown that
the logical complexity of the provability in LKID-omega is (Pi-1-1)-complete.
%\keywords{First-order logic, Inductive definitions, Truth predicate, Logical Complexity, Provability.}
\end{abstract}

\section{Introduction}
\label{sec:introduction}

Induction and recursion are essential principles for defining and computing, e.g., various sets. 
In programming languages, various data structures such as trees and lists are defined inductively, and computations on such structures are often defined recursively. 
A way to analyze these principles is to study the proof theory of systems endowed with induction and recursion mechanisms.
Recently, proof systems called LKID, \LKIDo\ and \CLKIDo\ \cite{PhDBrotherston,Brotherston2011}, which are based on sequent calculus, have attracted attention as such.
These are proof systems for first-order predicate language \FOLID\ with inductively defined predicates.

\LKIDo\ formalizes reasoning about inductive definitions by the infinite-descent method, 
and the proof tree may contain infinite branches.
On the other hand,
\CLKIDo\ is a system that only allows proof trees that are regular in the proof of \LKIDo, i.e., contain only a finite number of distinct subtrees.
In such proof trees, identical subtrees can be identified, making the overall proof figure cyclic.

Since these systems can formalize proofs of various properties of inductively defined predicates, it is natural to investigate the complexity of proof search for such systems.
The logical complexity of proof search by $\LKIDomega$ has received little attention, apart from a brief proof sketch in \cite{ito22}, largely because analyzing the complexity of infinite-descent proofs is challenging.
Besides their theoretical interests, logical systems of infinite-descent 
proofs are important because the unfolding of a cyclic proof becomes an 
infinite-descent proof, and infinite-descent proofs provide the basis for 
cyclic proofs.
In this paper we formally clarify the logical complexity of $\LKID^\omega$.

We will show the following three results in this paper:
(1) as a $\Pi^1_1$ formula we can define a truth predicate of 
a first-order language with inductive definitions in standard models,
(2) the validity of a first-order language with inductive definitions in standard models is
a $\Pi^1_1$ relation, and
(3) the provability in \LKIDo \ is $\Pi^1_1$-complete.

For result (1),
we will first show
the equivalence between the validity in countable standard models and
the validity in standard term models for
the signature extended by adding fresh constants.
Fresh constants serve as names for elements in the universe of
a given model.
These constants are used to show that
a term model is standard if a given model is standard.
Next, by using the downward Skolem-L\"{o}wenheim theorem,
we can improve this equivalence to
the validity in possibly uncountable standard models.
Then 
by extending the truth predicate of $\omega$-languages
given in Girard's book \cite{Girard}
with
the fact that inductive definitions can be coded as arithmetical formulas \cite{troelstrabook},
we will finally obtain our definition of a truth predicate of 
first order language with inductive definitions.

We can immediately prove result (2) since
the truth predicate obtained in result (1) is $\Pi^1_1$.

For result (3),
we will prove that the provability in \LKIDo\ is $\Pi^1_1$-complete
by showing that it is a $\Pi^1_1$ relation
and by showing that it is $\Pi^1_1$-hard.
The first claim is proved
by using 
the fact that its provability is equivalent to the truth in standard models \cite{Brotherston2011},
and 
result (2).
The second claim is proved
by showing the reduction from $\Pi^1_1$-hard problems to the truth of \LKIDo\
in standard models.

As mentioned above, our approach to defining the truth predicate for \FOLID\
is inspired by the method for defining the truth predicate for 
$\omega$-languages presented on page 348 of \cite{Girard}.
Since \FOLID\ may not be an $\omega$-language,
applying this idea to \FOLID\ requires establishing
the equivalence between validity in all models
and validity in all term models (part of result (1)).
The truth definition proposed in this paper may serve as a foundation for alternative 
approaches to truth definitions in higher-order languages beyond the class of 
$\omega$-languages.

The second author has been doing cooperative research with Stefano
Berardi for twenty years and they have written several papers on
program theory and mathematical logic, including a prize-awarded paper
on cyclic proofs.  The subject of this paper contains inductive
definitions, logical complexity, infinite-descent proofs, and cyclic
proofs, and it is strongly related to these activities.  We therefore
believe that the research activities reported in this paper fully fall
within Stefano Berardi's research interests and that this paper is an
appropriate means to honor him.

This paper is organized as follows:
Section \ref{FOLID} introduces the definition of \FOLID.
Section 3 establishes the equivalence between the validity in standard models and
the validity in standard term models.
Section 4 presents the truth predicate of \FOLID\ and shows that
the validity of inductive definitions is a $\Pi^1_1$ relation.
In Section 5, the infinite-descent proof system $\LKIDomega$ is defined.
In Section 6, the logical complexity of the provability in $\LKIDomega$ is shown to 
be $\Pi^1_1$-complete.
Section \ref{sec:related} reviews related work.
Finally, Section \ref{sec:conclusion} concludes with a discussion of an alternative proof technique for result (2) and outlines directions for future research.
\section{\FOLID: the first-order language with inductively defined predicates}
\label{FOLID}

This section provides 
backgrounds including
the definition of the syntax and semantics of the first-order language \FOLID\ with inductively defined predicates, as presented in \cite{Brotherston2011}.

The language of \FOLID\ consists of a countable language of first-order logic and finitely many \emph{inductive predicates}, which are distinct from ordinary predicate symbols.
In this section, we fix a signature $\Sigma$.
We write $c_1, c_2, \ldots$ for constant symbols, $f_1, f_2,\ldots$ for function symbols, $Q_1, Q_2, \ldots$ for ordinary predicate symbols, and $P_1,P_2,\ldots,P_n$ for inductive predicate symbols.
We assume each inductive predicate symbol has
its production rules defined below.
We refer to \emph{terms} and \emph{formulas} of $\Sigma$ when
they belong to the language generated by $\Sigma$.
We write $t(x_1,\ldots,x_m)$ for a term $t$, in which variables appearing in $t$ are included in $\{x_1,\ldots,x_m\}$.
We write $t[x:=u]$ for the term obtained from $t$ by replacing $x$ by $u$.
We also write $t(t_1,\ldots,t_m)$ for the term obtained by replacing $x_1,\ldots,x_m$  by $t_1,\ldots,t_m$ in $t$.
We use vector notations like $\bt$ and $\bx$ to represent sequences of terms and variables, respectively.
We may also use notation like $\bt(\bx)$ to clarify the variables $\bx$ contained in the sequence of terms $\bt$.
Then $\bt(\bu)$ means the sequence of terms obtained by replacing variables $\bx$ by $\bu$ in $\bt$.
The formulas of \FOLID\ are defined as those of traditional first-order logic with equality.

Function symbols and ordinary predicate symbols of \FOLID\ are interpreted by a structure $M=(U,\semb{ \ })$ as in traditional first-order logic.
Variables are interpreted as elements of the universe $U$ of $M$, by a variable assignment $\rho$.
We write $M \models_\rho F$ when a formula $F$ is true in $M$ and $\rho$.
When $F$ is closed, we simply write $M \models F$ to mean $M \models_\rho F$
for all $\rho$.
We provide the formal definition of syntax and semantics in Appendix \ref{app:A}.

The intended meaning of an inductive predicate symbol is specified by its \emph{production rules}.

\begin{Def}[Production Rules]\label{def:production_rules} \rm
\emph{Production rules of $P_i$} are of the following form:
\begin{equation}
 \infer{P_i(\Vec{t})}{Q_1(\Vec{u_1}) & \ldots & Q_h(\Vec{u_h}) & P_{j_1}(\Vec{t_1}) & \ldots & P_{j_m}(\Vec{t_m})} \label{eq:production}
\end{equation}
where $Q_1,\ldots,Q_h$ are ordinary predicate symbols,
$j_1,\ldots,j_m, i \in \{1,\ldots,n\}$ and vector symbols are sequences of terms whose lengths are the arities of the corresponding predicate symbols.
\end{Def}

This production rule means that
if 
$Q_1(\Vec{u_1})$, \ldots, $Q_h(\Vec{u_h})$,
$P_{j_1}(\Vec{t_1})$, \ldots, $P_{j_m}(\Vec{t_m})$ hold,
then
$P_i(\Vec{t})$ holds.
Since the assumption may contain the same predicate $P_i$ as the conclusion,
it can define $P_i$ inductively.
%Then the meaning of the predicate $P_i$ is defined as the least fixpoint
%of these production rules.

For simplicity, we may omit parentheses to write $Q \bu$ instead of $Q(\bu)$.

\begin{Eg} \label{ex:nat} \rm
 We define the production rules for the inductive predicate $N$ by the following:
\[
 \infer{N0}{} \quad \infer{Nsx}{Nx}
\]
where $s$ is a function symbol meaning ``successor function''.
The predicate $N$ denotes the set of ``natural numbers''.
\end{Eg}

We define a \emph{signature} of \FOLID\ as $(\Sigma,\Phi)$ where $\Phi$ is a finite set of production rules.
Next, we will define a standard model of $(\Sigma, \Phi)$.
For each inductive predicate $P_i$, a monotone operator is naturally defined on a universe $U$.
Let $k_i$ be the arity of the inductive predicate $P_i$.
Then, the corresponding monotone operator is defined as $\varphi_{i}:\cP(U^{k_1}) \times \cdots \times \cP(U^{k_n}) \to \cP(U^{k_i})$.
Intuitively, for given interpretations $X_1,\ldots,X_n$ of $P_1,\ldots,P_n$, a new interpretation of $P_i$ is determined by $\varphi_i(X_1,\ldots,X_n)$.

%\hrule

%We write
%$\rho(v_1,\ldots,v_l)$ for the sequence $(\rho(v_1),\ldots,\rho(v_l))$. 
% -> Move to just before Def 3.8

\begin{Def}[Operator for Production Rule]\rm
For a signature $\Sigma$, a set $\Phi$ of production rules,
a structure $M=(U, \semb{ \ })$ of $(\Sigma,\Phi)$,
if the inductive predicate symbols in $\Sigma$ are $P_1,\ldots,P_n$ of arity $k_1,\ldots,k_n$,
we define the operator $\varphi_i:\cP(U^{k_1}) \times \dots \times \cP(U^{k_n}) \to \cP(U^{k_i})$ 
for $P_i$ with $i \in \{1,\ldots,n\}$ as follows:
\begin{align*}
&\varphi_i(\Vec X)=\{ \Eval{\Vec{t}}\rho \ |\ 
%\\ \qquad
	\hbox{$\Phi$ has the production rule\ }
\hbox{$\vcenter{%
\infer{P_i \Vec t}{
	Q_1\Vec{u_1}
	&
	\ldots
	&
	Q_h\Vec{u_h}
	&
	P_{j_1} \Vec{t_1}
	&
	\ldots
	&
	P_{j_m} \Vec{t_m}
}%
}$},
\\ &\qquad
	\hbox{$\rho$ a variable assignment},
%\\ \qquad
	M \models_\rho Q_1 \Vec{u_1},
	\ldots,
	M \models_\rho Q_h \Vec{u_h},
%\\ \qquad
	\Eval{\Vec{t_1}}\rho \in X_{j_1},
	\ldots,
	\Eval{\Vec{t_m}}\rho \in X_{j_m}
\}.
\end{align*}
Finally, we define the operator $\varphi$ for $\Phi$ by
$
\varphi(\Vec X)=(\varphi_1(\Vec X),\ldots,\varphi_n(\Vec X)).
$
\end{Def}

%\hrule

For simplicity,
we will sometimes use $\subseteq$ for sequences
to denote the pointwise $\subseteq$-relation  on each elements
of sequences.
We call $(X_1,\ldots,X_n)$ that satisfies $\varphi(X_1,\ldots,X_n) \subseteq (X_1,\ldots,X_n)$ a \emph{prefixpoint} of $\varphi$.
It is well-known that the least prefixpoint 
is the least fixpoint.
We write $\Lfp.\varphi$ for the least prefixpoint of $\varphi$.

\begin{Def}[Standard Model]\rm
A first-order structure $M$ for $\Sigma$ is said to be a
\emph{standard model} for 
$(\Sigma,\Phi)$ if $\Eval{P_i} = (\Lfp.\varphi)_i$ for all $i \in \{1,\ldots,n\}$ where $\varphi$ is the monotone operator
for $\Phi$.
\end{Def}

{\bf Analytical hierarchy.}
The analytical hierarchy is a classification of relations on $\Nat$
according to the complexity of second-order logical formulas that
define relations. 
We consider the following second-order formulas $Q_1 X_1
\ldots Q_n X_n \phi$ where $Q_i$ is $\exists$ or $\forall$, each $X_i$
is a second-order variable, and $\phi$ is a formula 
that does not contain second-order quantifiers.
Then $\Pi^1_1$ relations are defined as
relations described by those formulas where all $Q_i$'s are $\forall$.
We call these formulas $\Pi^1_1$ formulas.

\section{Term Models}

This section shows the equivalence between the validity of inductive definitions
in standard models and
the validity of inductive definitions in standard term models of the extended signature
obtained by adding fresh constants.

In order to define a truth predicate,
we need a term model from a given model.
To construct a standard term model,
we will introduce a notion of name extension of a structure.

\begin{Def}[Name-Extended Model]\rm
For a signature $\Sigma$, we define a signature $\Sigma_c$ as
$\Sigma \cup \{ c_1,c_2,\ldots \}$ where
$c_1,c_2,\ldots$ are fresh constants.
We call them {\em name constants}.
A model $(U,\Eval\ )$ of $\Sigma_c$ is defined to be {\em name-extended}
if for any $u \in U$
there is some $c_i$ such that $\Eval{c_i}=u$.
\end{Def}

Note that this definition is slightly different from the well-known
structure expansion by names, as the signature is first extended by a
countable number of fresh constants, and a structure of that signature
is defined as name-extended if every element of its universe
has some constant whose interpretation is the element.

\begin{Def}[Model $M_c$]\rm
For a countable structure $M = (U,\Eval\ )$ of $(\Sigma,\Phi)$,
we define a structure $M_c=(U,\Eval\ _c)$ of $(\Sigma_c,\Phi)$ by
$\Eval{c_i}_c=u_i$ for any $i$ where $U=\{u_1,u_2,\ldots\}$ 
(if $U$ is finite and $|U|=n$ then we define $\Eval{c_i}_c=u_1$ for $i \ge n$).
\end{Def}

The next lemma says that the model $M_c$ 
constructed in the above way is name-extended.

\begin{Lemma}
$M_c$ is name-extended.
\end{Lemma}

\begin{proof}
For any $u \in U$, if $u=u_i$ then $\Eval{c_i}_c=u_i=u$.
\end{proof}

\vspace{0.5\baselineskip}

The next lemma says that 
if a model is standard, its name-extended model
is also standard.
It clearly holds since
a production rule does not contain name constants
and 
$M_c \models$ used in the definition of the operator
for production rules
is the same as $M \models$ used in that.

\begin{Lemma}\label{lemma:name-standard}
If $M$ is a standard model, 
$M_c$ is a standard model.
\end{Lemma}

\begin{Lemma}\label{lemma:1a}\label{lemma:conserv-name}
For a model $M$ of $\Sigma$ and
a closed formula $A$ of $\Sigma$,
$M \models A$ iff $M_c \models A$.
\end{Lemma}

\begin{proof}
The universes for $M$ and $M_c$ are the same,
and the interpretations for $M$ and $M_c$ are also the same
for formulas of $\Sigma$.
\end{proof}

\begin{Def}[Term Model]\rm
We call a structure $(U,\Eval\ )$ of $\Sigma$
a {\em term model} if

(1) $U = S/\sim$, where 
$S$ is the set of closed terms of $\Sigma$ and
the relation $\sim$ is an equivalence relation on $S$,

(2)
$\Eval c=[c]$ for any constant $c$, and
$\Eval f([\Vec u])=[f(\Vec u)]$ for any function symbol $f$,
where $[t]$ denotes the equivalence class of $t$.
\end{Def}

We will write $[t]$ for the equivalence class of $t$.

\begin{Def}[Model $M_T$]\rm
For a structure $M=(U,\Eval\ )$ of $(\Sigma,\Phi)$,
we define a term model 
$M_T=(U_T,\Eval\ _T)$ of $(\Sigma,\Phi)$
as follows:
\begin{align*}
&t \sim u \hbox{\ for closed terms $t,u$ if\ } M \models t=u, \\ \relax
&[\Vec u] \in \Eval {P_i}_T
\hbox{\ for an inductive predicate symbol $P_i$ if\ } 
M \models P_i(\Vec u), \\ \relax
&[\Vec u] \in \Eval Q_T
\hbox{\ for an ordinary predicate symbol $Q$ if\ } M \models Q(\Vec u).
\end{align*}
\end{Def}

\begin{Lemma}\label{lemma:term-name}
If $M$ of $\Sigma_c$ is name-extended, 
$M_T$ of $\Sigma_c$ is also name-extended.
\end{Lemma}

\begin{proof}
Assume $u \in U_T$ in order to find $c_i$ such that $\Eval{c_i}_T=u$.
Then, there is a closed term $t$ such that $u=[t]$.
Take any variable assignment $\rho$.
Then $\Eval t\rho=u'$ for some $u' \in U$.
Since $M$ is name-extended, there is some $c_i$ such that $\Eval{c_i}=u'$.
Then $\Eval{c_i}_T=[c_i]$.
Hence $M \models t=c_i$.
Hence $[t]=[c_i]$.
Hence $\Eval{c_i}_T=u$.
\end{proof}

%\hrule
%
\vspace{0.5\baselineskip}

For easy reading,
we sometimes write $\Vecc{e_i}i$ for a sequence $e_1,\ldots,e_k$ by
explicitly describing the index $i$.

\begin{Def}
For an inductive predicate symbol $P$ and a number $k$,
we define a formula $P^{(k)}(\Vec x)$ by
\begin{align*}
&P^{(0)}(\Vec x) \equiv \bottom, \\
&P^{(k+1)}(\Vec x) \equiv \Lor\{
\exists\Vec y(\Vec x=\Vec {t_0} \land \Vecc{Q_l\Vec {u_l}}l \land \Vecc{P_{j_p}^{(k)}(\Vec {t_p})}p ) \ |\
\\
&\qquad
\vcenter{\infer{P\Vec {t_0}}{\Vecc{Q_l\Vec {u_l}}l & \Vecc{P_{j_p}\Vec{t_p}}p}} 
\hbox{\ a production rule},
\Vec y \hbox{\ the free variables of the production rule}
\}.
\end{align*}
\end{Def}

This definition introduces a notion of approximation for fixpoints.

\begin{Prop}\label{prop:unfold}
For a model $M = (U, \Eval{\ })$ of $\Sigma$,
$\Eval{\Vec t}\rho \in (\varphi^k(\Vec\emptyset))_i$
iff
$M \models_\rho P_i^{(k)}(\Vec t)$ for all terms $\Vec t$.
\end{Prop}

\begin{proof}
$\Longrightarrow$: By induction on $k$.
Assume $\Eval{\Vec t}\rho \in (\varphi^{k+1}(\Vec\emptyset))_i$ to show
$M \models_\rho P_i^{(k+1)}(\Vec t)$.
Then, we have 
$
\Eval{\Vec t}\rho = \Eval{\Vec {t_0}}\rho',% \\
M \models_{\rho'} Q_l\Vec {u_l},% \\
\Eval{\Vec {t_p}}\rho' \in (\varphi^k(\Vec\emptyset))_{j_p}
$
$(1 \le l \le h, 1 \le p \le m)$
for some $\rho'$ and some production rule
\[
\infer{P_i \Vec {t_0}}{\Vecc{Q_l\Vec {u_l}}l & \Vecc{P_{j_p}\Vec{t_p}}p}
\quad (1 \le l \le h, 1 \le p \le m).
\]

By IH, we have
$
M \models_{\rho'} P_{j_p}^{(k)}(\Vec {t_p}).
$
Hence
$
M \models_{\rho'} P_i^{(k+1)}(\Vec {t_0}).
$
Hence
$
M \models_{\rho} P_i^{(k+1)}(\Vec t).
$

$\Longleftarrow$: By induction on $k$.
Assume $M \models_\rho P_i^{(k+1)}(\Vec t)$ to show
$\Eval{\Vec t}\rho \in (\varphi^{k+1}(\Vec\emptyset))_i$.
Then
$
M \models_{\rho} \exists\Vec y(\Vec t=\Vec {t_0} \land \Vecc{Q_l\Vec {u_l}}l \land \Vecc{P_{j_p}^{(k)}(\Vec {t_p})}p)
$
for some production rule
\[
\infer{P_i \Vec {t_0}}{\Vecc{Q_l\Vec {u_l}}l & \Vecc{P_{j_p}\Vec{t_p}}p}.
\]

Hence, there is some $\rho'=\rho[\Vec y := \Vec u]$ such that
$
M \models_{\rho'} 
\Vec t=\Vec {t_0} \land \Vecc{Q_l\Vec {u_l}}l \land \Vecc{P_{j_p}^{(k)}(\Vec {t_p})}p.
$
Hence
\[
\Eval{\Vec t}\rho'=\Eval{\Vec {t_0}}\rho',% \\
M \models_{\rho'} \Vecc{Q_l\Vec {u_l}}l,% \\
M \models_{\rho'} \Vecc{P_{j_p}^{(k)}(\Vec {t_p})}p.
\]
By IH, we have
$
\Eval{\Vec {t_p}}\rho' \in (\varphi^k(\Vec\emptyset))_{j_p}.
$
Hence
$
\Eval{\Vec {t_0}}\rho' \in (\varphi^{k+1}(\Vec\emptyset))_i.
$
Hence
$
\Eval{\Vec t}\rho \in (\varphi^{k+1}(\Vec\emptyset))_i.
$
\end{proof}

\begin{Cor}\label{cor:unfold}
For a structure $M$, the following are equivalent:

(1) $M$ is standard.

(2) For any inductive predicate $P$,
$
M \models_\rho P^{(k)}(\Vec t)
$
for some $k$
iff
$
M \models_\rho P(\Vec t).
$
\end{Cor}

%\hrule

The next lemma says that 
if a model is name-extended and standard, its term model
is also standard.

In the following, for a variable assignment $\rho$ from variables to 
closed terms, and a term $t$,
we write $\rho t$ for $t[\Vec x:=\Vec{\rho(x)}]$, where
$\Vec x$ is the free variables of $t$.

\begin{Lemma}\label{lemma:term-standard}
If $M$ is a name-extended standard model, 
$M_T$ is a standard model.
\end{Lemma}

\begin{proof}
Let $M = (U, \Eval{\ })$ be a name-extended standard model of $(\Sigma,\Phi)$.
Let $\varphi_T$ be the operator for production rule of $M_T$.

First, we will show
$
\Eval{P_i}_T \supseteq (\bigcup_k \varphi_T^k(\Vec\emptyset))_i
$.
For this,
we will show
$
\Eval{P_i}_T \supseteq (\varphi_T^k(\Vec\emptyset))_i
$
by induction on $k$.
Assume
$
[\Vec {t_0}] \in (\varphi_T^{k+1}(\Vec\emptyset))_i
$
to show
$
[\Vec {t_0}] \in \Eval{P_i}_T
$.
Then, there are a variable assignment $\rho$ on $M_T$ and 
some production rule 
\[
\infer{P_i \Vec t}{
	Q_1\Vec{u_1}
	&
	\ldots
	&
	Q_h\Vec{u_h}
	&
	P_{j_1} \Vec{t_1}
	&
	\ldots
	&
	P_{j_m} \Vec{t_m}
}
\]
such that
$[\Vec {t_0}] = \Eval{\Vec t}_T\rho$ and
$
	M_T \models_\rho Q_1 \Vec{u_1},
	\ldots,
	M_T \models_\rho Q_h \Vec{u_h},
%\\ \qquad
	\Eval{\Vec{t_1}}_T\rho \in (\varphi_T^{k}(\Vec\emptyset))_{j_1},
	\ldots,
	\Eval{\Vec{t_m}}_T\rho \in (\varphi_T^{k}(\Vec\emptyset))_{j_m}.
$
Define a variable assignment $\rho'$ on $U$ by
$\rho'(x)=\Eval u$ if $\rho(x)=[u]$.
Define a variable assignment $\rho''$ on closed terms by
$\rho''(x)=u$ if $\rho(x)=[u]$.

Then
$[\Vec {t_0}] = [\rho''(\Vec t)]$ %,
%$
%	M \models_{\rho'} Q_l \Vec{u_l} \ (1 \le l \le h),% \\ \relax
%$
%\\ \qquad
and
$
	[\rho''(\Vec{t_p})] \in (\varphi_T^{k}(\Vec\emptyset))_{j_p}
	\ (1 \le p \le m)
$.

By IH, we have
$
	[\rho''(\Vec{t_p})] \in \Eval{P_{j_p}}_T.
$
Since $\Eval{\rho''(\Vec{t_p})}=\Eval{\Vec{t_p}}\rho'$, we have
$
	\Eval{\Vec{t_p}}\rho' \in \Eval{P_{j_p}}.% \\
$
Hence $M \models_{\rho'} P_{j_p}(\Vec{t_p})$.

For $Q_l$ \ $(1 \le l \le h)$,
from $M_T \models_\rho Q_l \Vec{u_l}$
we have $\Eval{\Vec{u_l}}_T\rho \in \Eval{Q_l}_T$.	
Since $\Eval{\Vec{u_l}}_T\rho=[\rho''(\Vec{u_l})]$,
we have $\Eval{\rho''(\Vec{u_l})} \in \Eval{Q_l}$.
Since $\Eval{\rho''(\Vec{u_l})}=\Eval{\Vec{u_l}}\rho'$,
we have $M \models_{\rho'} Q_l\Vec{u_l}$.

Since $M$ is standard, we have
$
\Eval{\Vec t}\rho' \in \Eval{P_i}.
$
Hence
$
[\rho''(\Vec t)] \in \Eval{P_i}_T.
$
Hence
$
[\Vec {t_0}] \in \Eval{P_i}_T.
$

Secondly, we will show
$
\Eval{P_i}_T \subseteq (\bigcup_k \varphi_T^k(\Vec\emptyset))_i.
$
Assume $[\Vec {t_0}] \in \Eval{P_i}_T$
to show
$[\Vec {t_0}] \in (\bigcup_k \varphi_T^k(\Vec\emptyset))_i$.
Then
$
M \models P_i\Vec {t_0}.
$
Since $M$ is standard, for some $k$ we have
$
\Eval{\Vec {t_0}} \in (\varphi^{k+1}(\Vec\emptyset))_i.
$
Hence there are some $\rho$ on $U$ and some production rule 
\[
\infer{P_i\Vec t}{\Vecc{Q_l\Vec {u_l}}l & \Vecc{P_{j_p}\Vec {t_p}}p}
\quad (1 \le l \le h, 1 \le p \le m)
\]
such that
$
\Eval{\Vec {t_0}}=\Eval{\Vec t}\rho,% \\
M \models_\rho Q_l\Vec {u_l},% \\
\Eval{\Vec {t_p}}\rho \in (\varphi^{k}(\Vec\emptyset))_{j_p}
\text{ for all } 1 \le l \le h, 1 \le p \le m.
$

By Proposition \ref{prop:unfold},
$
M \models_\rho P_{j_p}^{(k)}(\Vec {t_p})
\text{ for all } 1 \le p \le m.
$

Hence
$
M \models_\rho \Vec {t_0} = \Vec t \land \Vecc{Q_l\Vec {u_l}}l \land \Vecc{P_{j_p}^{(k)}(\Vec {t_p})}p.
$

Since $M$ is name-extended,
for each $i$ there is some $c_i$ such that
$\rho(x_i)=\Eval{c_i}$.
Let $\Vec c$ is the sequence of $c_i's$.
Then
$
M \models 
(\Vec {t_0} = \Vec t \land \Vecc{Q_l\Vec {u_l}}l \land \Vecc{P_{j_p}^{(k)}(\Vec {t_p})}p)[\Vec x:=\Vec c].
$

Define a variable assignment $\rho'$ on 
the universe of $M_T$ by
$
\rho'(\Vec x)=[\Vec c].
$
Then
$
[\Vec {t_0}] = \Eval {\Vec t}_T\rho',% \\
M_T \models_{\rho'} Q_l \Vec {u_l},% \\
M_T \models_{\rho'} P_{j_p}^{(k)}(\Vec {t_p})
\text{ for all } 1 \le l \le h, 1 \le p \le m.
$
By Proposition \ref{prop:unfold},
$
\Eval{\Vec {t_p}}_T\rho' \in (\varphi^k_T(\Vec\emptyset))_{j_p}
\text{ for all } 1 \le p \le m.
$
Hence
$
\Eval{\Vec t}_T\rho' \in (\varphi^{k+1}_T(\Vec\emptyset))_i.
$
Since $\Eval{{\Vec t}}_T\rho' = [\Vec {t_0}]$, we have
$
[\Vec {t_0}] \in (\varphi^{k+1}_T(\Vec\emptyset))_i.
$
Hence
$[\Vec {t_0}] \in (\bigcup_k \varphi_T^k(\Vec\emptyset))_i$.
\end{proof}

\begin{Lemma}\label{lemma:1b}\label{lemma:equiv-term}
For a name-extended model $M$ of $\Sigma_c$ and
a closed formula $A$ of $\Sigma_c$,
$M \models A$ iff $M_T \models A$.
\end{Lemma}

\begin{proof}
By induction on $A$. We will show only a difficult case.

Case $A \equiv \forall xB$.
$\Longleftarrow$:
Assume $M_T \models \forall xB$ in order to show
$M \models \forall xB$.
Assume $M \not\models \forall xB$ in order to show contradiction.
Hence $M \models \exists x\neg B$.
Since $M$ is name-extended,
there is some $c_i$ such that
$M \models \neg B[x:=c_i]$.
By IH, $M_T \models \neg B[x:=c_i]$.
Hence $M_T \models \exists x\neg B$.
Hence, $M_T \not\models \forall x B$, which leads to a contradiction.

$\Longrightarrow$: This case can be similarly shown using Lemma \ref{lemma:term-name}.
\end{proof}

\vspace{0.5\baselineskip}

We write $M_{cT}$ for $(M_c)_T$.

From Lemmas \ref{lemma:1a} and \ref{lemma:1b},
we have the equivalence between the validity in $M$ and the validity in $M_{cT}$ for closed formulas of $\Sigma$
in the next lemma.

\begin{Lemma}\label{lemma:1}
For a model $M$ of $\Sigma$ and
a closed formula $A$ of $\Sigma$,
$M \models A$ iff $M_{cT} \models A$.
\end{Lemma}

\begin{Prop}\label{prop:countablestandard}
If $M$ is a standard model of $(\Sigma,\Phi)$,
then there exists a countable standard model $M'$ of $(\Sigma,\Phi)$
that is elementarily equivalent to $M$.
\end{Prop}

\begin{proof}
Assume $M=(U,\Eval\ )$ is a standard model
in order to construct a countable standard model 
$M'$ of $(\Sigma,\Phi)$
that is elementarily equivalent to $M$.

By the downward Skolem-L\"owenheim theorem (Theorem 3.3.12 in \cite{logicandstructure}),
there is a countable structure $M'=(U', \Eval\ ')$ which is 
an elementary substructure of $M$.

We will first show that
$
M' \models_\rho P(\Vec t)
$
iff
$
M' \models_\rho P^{(k)}(\Vec t)
$
for some $k$.

For the only-if-part.
Assume
$M' \models_\rho P(\Vec t)$.
Let $FV(\Vec t)$ be $\Vec x$.
Let $\rho(\Vec x)=\Vec u$.
Let $\Vec{\Bar u}$ be names of $\Vec u$.
Then
$M' \models P(\Vec t)[\Vec x:=\Vec{\Bar u}]$
where
$M' \models$ is extended to the signature $\Sigma \cup \{ \Bar u \ |\ u \in U'\}$
by $\Eval{\Bar u}'=u$.
Since $M'$ is an elementary substructure of $M$,
we have
$M \models P(\Vec t)[\Vec x:=\Vec{\Bar u}]$
where
$M \models$ is extended to the signature $\Sigma \cup \{ \Bar u \ |\ u \in U\}$
by $\Eval{\Bar u}=u$.
Since $M$ is standard,
by Corollary \ref{cor:unfold}, 
$M \models P^{(k)}(\Vec t)[\Vec x:=\Vec{\Bar u}]$
for some $k$.
Since $M'$ is an elementary substructure of $M$,
we have
$M' \models P^{(k)}(\Vec t)[\Vec x:=\Vec{\Bar u}]$.
Hence
$M' \models_\rho P^{(k)}(\Vec t)$.

The if-part is proved in a similar way to the only-if-part.

By Corollary \ref{cor:unfold}, $M'$ is standard.
\end{proof}

\begin{Prop}\label{prop:equiv-term}
For $(\Sigma,\Phi)$ and a closed formula $A$ of $\Sigma$,
the following are equivalent:

(1) $M \models A$ for every standard model $M$ of $(\Sigma,\Phi)$.

(2) $M \models A$ for every standard term model $M$ of $(\Sigma_c,\Phi)$.
\end{Prop}

\begin{proof}
The claim from (1) to (2) clearly holds.

We will show the claim from (2) to (1).
Assume (2) and
fix a structure $M$ of $(\Sigma,\Phi)$ in order to show $M \models A$.
By Proposition \ref{prop:countablestandard},
there is a countable standard model $M'$ of $(\Sigma,\Phi)$
that is elementarily equivalent to $M$.
By Lemmas \ref{lemma:name-standard} and \ref{lemma:term-standard},
$M'_{cT}$ is a standard term model of $(\Sigma_c,\Phi)$.
By (2), $M'_{cT} \models A$.
By Lemma \ref{lemma:1},
$M' \models A$.
Since $M$ and $M'$ are elementarily equivalent, we have
$M \models A$.
\end{proof}

\section{Truth Predicate of Inductive Definitions}

In this section,
we will give the truth predicate of \FOLID\ and
will show that the validity of inductive definitions is a $\Pi^1_1$ relation.

\subsection{Coding of Inductive Predicates}

In this section,
we present some coding of inductive definitions in arithmetic
by applying the idea of
arithmetical representation of 
inductive definitions (Theorem 1.4.5 in \cite{troelstrabook}) 
to inductive definitions of \FOLID.

We write $(\bz)_i$ for the $i$-th element of a sequence $\bz$ and
$|\bz|$ for the length of $\bz$, where $i$ starts with $0$.
We write $\Code e$ for the G\"{o}del coding of the expression $e$.

First, we define an operator for codes that
corresponds to the operator of inductive predicates.

\begin{Def}[Code Operator for Production Rules]\rm
For a signature $(\Sigma,\Phi)$ and 
a function variable $f$,
if the inductive predicate symbols in $\Phi$ are $P_1,\ldots,P_n$,
we define $\tilde\varphi_i:\mathcal{P}(\Nat)^n \to \mathcal{P}(\Nat)$ by
\begin{align*}
&\tilde\varphi_i(\Vec X)=\{ \Code{\rho\Vec{t}} \ |\ 
\\& \qquad
	\hbox{$\Phi$ has the production rule\ }
\hbox{$\vcenter{
\infer{P_i \Vec t}{
	Q_1\Vec{u_1}
	&
	\ldots
	&
	Q_h\Vec{u_h}
	&
	P_{j_1} \Vec{t_1}
	&
	\ldots
	&
	P_{j_m} \Vec{t_m}
}
}$},
\\& \qquad
	\hbox{$\rho$ a variable assignment from variables to closed terms of $\Sigma$},
%\\ \qquad
	f(\Code{Q_1\rho \Vec{u_1}})=0,
	\ldots,
\\& \qquad
	f(\Code{Q_h\rho \Vec{u_h}})=0,
	\Code{\rho \Vec{t_1}} \in X_{j_1},
	\ldots,
	\Code{\rho \Vec{t_m}} \in X_{j_m}
\}, \\&
\tilde\varphi(\Vec X)=(\tilde\varphi_1(\Vec X),\ldots,\tilde\varphi_n(\Vec X)).
\end{align*}
\end{Def}

\begin{Lemma}\label{lemma:2}
For a term model $M$ of $(\Sigma,\Phi)$ and a function variable $f$,
if $M \models Q_i\Vec u$ is equivalent to $f(\Code{Q_i\Vec u})=0$
for all $Q_i$ and closed terms $\Vec u$,
then we have
$[\Vec t] \in (\bigcup_k \varphi^k(\Vec\emptyset))_i$
$\Longleftrightarrow$
$\Code{\Vec t} \in (\bigcup_k \tilde\varphi^k(\Vec\emptyset))_i$
for all closed terms $\Vec t$.
\end{Lemma}

\begin{proof}
By induction on $k$ with the definition of $\varphi$ and $\tilde\varphi$,
we can show that
$[\Vec t] \in (\varphi^k(\Vec\emptyset))_i$
iff
$\Code{\Vec t} \in (\tilde\varphi^k(\Vec\emptyset))_i$.
The claim follows from it.
\end{proof}

\vspace{0.5\baselineskip}

We assume some coding in Peano arithmetic
and write $\<a_1,\ldots,a_k\>$ for the code of
a sequence $a_1,\ldots,a_k$ of numbers.
We write $(\<a_1,\ldots,a_k\>)_i$ for $a_{i+1}$.
We write $a \tilde\in b$ if $(b)_i=a$ for some $i$.
For sequences $a_1,\ldots,a_k$ of numbers,
we write $\Bar{(a_1,\ldots,a_k)}$ for 
the sequence $(S_1,\ldots,S_k)$,
where $S_i$ is the set $\{ b \ |\ b \tilde\in a_i\}$.

A finite variable assignment is a finite function
from some finite subset of variables to the universe.

\begin{Def}\rm
For $(\Sigma,\Phi)$ and a function variable $f$,
if the inductive predicate symbols in $\Phi$ are $P_1,\ldots,P_n$,
we define a formula $W(y,z)$ of Peano arithmetic
with a function variable $f$ by
\begin{align*}
&W(y,z) \equiv
\forall i(1 \le i \le n \imp
\forall x \tilde\in (z)_i(
\\& \qquad
	\hbox{$\Phi$ has some production rule\ }
\hbox{$\vcenter{
\infer{P_i \Vec t}{
	Q_1\Vec{u_1}
	&
	\ldots
	&
	Q_h\Vec{u_h}
	&
	P_{j_1} \Vec{t_1}
	&
	\ldots
	&
	P_{j_m} \Vec{t_m}
}
}$},
\\& \qquad
	\parbox{12cm}{and there is some finite variable assignment $\rho$ 
from the free variables of the production rule to closed terms of $\Sigma$, and}
\\& \qquad
	f(\Code{Q_1\rho \Vec{u_1}})=0 \land
	\ldots \land
	f(\Code{Q_h\rho \Vec{u_h}})=0 \land
%\\ \qquad
	\Code{\rho \Vec{t_1}} \tilde\in (y)_{j_1} \land
	\ldots \land
	\Code{\rho \Vec{t_m}} \tilde\in (y)_{j_m} \land
\\& \qquad
x=\Code{\rho\Vec{t}})).
\end{align*}
\end{Def}

\begin{Lemma}\label{lemma:3}
$a \in (\tilde\varphi^k(\Vec\emptyset))_i$ iff
the following formula is true in the standard model of arithmetic:
$
\exists z(|z|=k+1 \land (z)_0=\Vec{\<\ \>}
\land \forall l<k\ W((z)_l,(z)_{l+1})
\land a \tilde\in ((z)_k)_i).
$
\end{Lemma}

\begin{proof}
By induction on $k$.

Case $k=0$. The claim holds since both sides of the claim are false.

Case $k>0$. We show only a difficult case. $\Longrightarrow$:

Assume $a \in (\tilde\varphi^k(\Vec\emptyset))_i$.
Then $a \in (\tilde\varphi(\tilde\varphi^{k-1}(\Vec\emptyset)))_i$.
Hence there are some production rule
\[
\infer{P_i \Vec t}{
	Q_1\Vec{u_1}
	&
	\ldots
	&
	Q_h\Vec{u_h}
	&
	P_{j_1} \Vec{t_1}
	&
	\ldots
	&
	P_{j_m} \Vec{t_m}
}
\]
and 
some variable assignment $\rho$ on closed terms
such that
$\Code{\rho \Vec t} = a$
and
$f(\Code{Q_l\rho\Vec {u_l}})=0 \ (1 \le l \le h)$
and
$\Code{\rho\Vec {t_l}} \in (\tilde\varphi^{k-1}(\Vec\emptyset))_{j_l}
\ (1 \le l \le m)$.
For each $l \le m$, by IH, there is some $w_l$ such that
$
|w_l|=k \land (w_l)_0=\Vec{\<\ \>}
\land \forall p<k-1\ W((w_l)_p,(w_l)_{p+1})
\land \Code{\rho\Vec{t_l}} \tilde\in ((w_l)_{k-1})_{j_l}.
$

Define a sequence $z$ of $n$-sequences of sequences of numbers by
\[
(\Bar{(z)_p})_q = \bigcup_{l \le m}(\Bar{(w_l)_p})_q \ (q \le n, p < k), \\
(\Bar{(z)_k})_q = \{\Code{\rho\Vec t}\} \ (q = i), \\
(\Bar{(z)_k})_q = \emptyset \ (q \ne i).
\]
Then
$
|z|=k+1 \land (z)_0=\Vec{\<\ \>}
\land \forall l<k \ W((z)_l,(z)_{l+1})
\land a \tilde\in ((z)_k)_i.
$
\end{proof}

\begin{Def}\label{def:tildeP}\rm
$\widetilde{P_i}$ is defined by the following formula of Peano arithmetic with a function variable $f$:
$
\tilde P_i(a) \equiv
\exists z((z)_0=\Vec{\<\ \>}
\land \forall l<|z|-1\ W((z)_l,(z)_{l+1})
\land a \tilde \in ((z)_{|z|-1})_i).
$
\end{Def}

We explain how our definition of $\tilde P_i$ is obtained.
First, we code
$[\Vec t] \in (\bigcup_k\varphi^k(\Vec\emptyset))_i$
by
$\Code{\Vec t} \in (\bigcup_k\tilde\varphi^k(\Vec\emptyset))_i$.
Next,
$\Code{\Vec t} \in (\tilde\varphi^k(\Vec\emptyset))_i$ iff
there are finite sets $\Vec {S_0}, \Vec {S_1},\ldots, \Vec {S_k}$ such that
$\tilde\varphi(\Vec {S_l}) \supseteq \Vec{S_{l+1}}$ for $l < k$ and
$(\Vec {S_k})_i \ni \Code{\Vec t}$.
Finally,
we code each finite set $\Vec {S_l}$ by some list of numbers,
and we code the sequence of these lists by a number $z$.

By Lemmas \ref{lemma:2} and \ref{lemma:3},
we have the following equivalence:

\begin{Lemma}\label{lemma:indpred}
For a term model $M$ of $(\Sigma,\Phi)$ and a closed term $t$ of $\Sigma$
and a function variable $f$,
if $M \models Q_i\Vec u$ is equivalent to $f(\Code{Q_i\Vec u})=0$,
then we have the following:
$\Tilde P_i(\Code{\Vec t})$ iff 
$[\Vec t] \in (\bigcup_k \varphi^k(\Vec\emptyset))_i$.
\end{Lemma}

\subsection{Truth Predicate}

In this section,
we will define a truth predicate for \FOLID\ formulas,
by extending to \FOLID\
the truth predicate for $\omega$-languages
given in the proof of Theorem 6.1.4 in \cite{Girard}.

For a formula $B$ of $\Sigma$,
we use an abbreviation
$\forall \Code B (\_ \_ \Code{\ldots B \ldots} \_ \_)$
to denote 
$\forall n ({\rm Formula}(n,\Code\Sigma) \imp \_ \_ {\rm Subst}(\Code{\ldots x \ldots},x,n) \_ \_)$,
where ${\rm Formula}(n,\Code\Sigma)$ means that 
$n$ is a code of some closed formula 
in the first-order language with its signature $\Sigma$,
and
${\rm Subst}(x,y,z)$ means the code of $e[v:=e']$ when
$x$ is the code of some expression $e$,
$y$ is the code of some variable $v$,
and
$z$ is the code of some expression $e'$.
In a similar way,
we also use an abbreviation
$\forall \Code x$ for variables,
$\forall \Code t$ for closed terms, and
$\forall \Code {\Vec t}$ for sequences of closed terms.
We also use an abbreviation for $\exists\Code{\ldots}$
in a similar way to $\forall\Code{\ldots}$.
We write $\arity(v)$ for the arity of a function symbol or a predicate symbol 
$v$.

In the next formula $I(f)$,
$f$ is the truth predicate and
$f(\Code A) = 0$ means that $A$ is true.

\begin{Def}\label{def:I}\rm
For $(\Sigma,\Phi)$,
we define a formula $I(f)$ as follows,
where $f$ is a function variable,
$P_1,\ldots,P_n$ are inductive predicate symbols in $\Sigma$,
$B$ and $C$ range over closed formulas of $\Sigma_c$,
and
$t$ and $u$ range over closed terms of $\Sigma_c$.
\begin{align*}
 &\forall \Code B (f(\Code B) = 0 \lor f(\Code B)=1) \wedge \\
 &\forall \Code B (f(\Code {\neg B}) = 1 - f(\Code B)) \wedge \\
 &\forall \Code B \forall \Code C (f(\Code{B \wedge C}) = \max \{f(\Code B), f(\Code C)\}) \wedge \\
 &\forall \Code B \forall \Code C (f(\Code{B \vee C}) = \min\{f(\Code B), f(\Code C)\}) \wedge \\
 &\forall \Code B \forall \Code C (f(\Code{B \to C}) = \min\{1 - f(\Code B), f(\Code C)\}) \wedge \\
 &\forall \Code B (f(\Code{\forall x B(x)}) = 0 \leftrightarrow \forall \Code t f(\Code{B(t)}) = 0) \wedge \\
 &\forall \Code B (f(\Code{\exists x B(x)}) = 0 \leftrightarrow \exists \Code t f(\Code{B(t)}) = 0) \wedge \\
 &\forall \Code t (f(\Code{t = t}) = 0) \wedge \\
 &\forall \Code t \forall \Code u \forall \Code x \forall \Code B (f(\Code{t=u \wedge B[x:=t] \to B[x:=u]})=0) \wedge \\
 &\bigwedge_{i \in [1,n]} \forall \Code \bt (\arity(P_i)=|\bt|\to 
(f(\Code{P_i(\bt)})=0 \lequiv \widetilde{P_i}(\Code{\bt}))).
\end{align*}
\end{Def}

The next proposition is a key for the truth predicate.
It says that the validity in standard term models is equivalent
to the value of the truth predicate.

\begin{Prop} \label{prop:henkin}
For $(\Sigma,\Phi)$ and a closed formula $A$ of $\Sigma_c$, 
the following are equivalent:

(1) $M \models A$ for any standard term model $M$ of $(\Sigma_c,\Phi)$.

(2) $\forall f (I(f) \to f(\Code{A})=0)$ is true.
\end{Prop}

\begin{proof}
The direction $(2) \Rightarrow (1)$:
Fix a standard term model $M = (U, [\![ \ ]\!])$, where
$U$ is the universe of $M$, and
$[\![ \ ]\!]$ is the interpretation of non-logical symbols.
We define $f$ by
$f(\Code B)=0$ if $M \models B$ and $f(\Code{B})=1$ if $M \not \models B$.

We will show $I(f)$.
For this we will show each conjunct of it.
We only show difficult cases.

Conjunct for $\forall x B(x)$:
First
we show the following claim:

- $M \models B(t)$ for any closed term $t$ implies
$M \models \forall x B(x)$.

Assume $M \not \models \forall x B(x)$ in order to show contradiction.
Then we have $M \models \exists x \neg B(x)$.
Since $M$ is a term model, there is some closed term $t$
such that $M \models \neg B(t)$.
This contradicts $M \models B(t)$ for any closed term $t$.
We have shown the claim.

Then,
$f(\Code{\forall x B(x)}) = 0 \Leftrightarrow 
M \models \forall x B(x) \Leftrightarrow \hbox{\ (by claim)\ }
M \models B(t) \hbox{\ for any closed term $t$} 
\Leftrightarrow \hbox{ (by definition of $f$) }
\forall \Code t (f(\Code{B(t)}) = 0)$.

Conjunct for $P_i\Vec t$:
We will show that
$f(\Code{P_i\Vec t})=0$ iff $\tilde P_i(\Code{\Vec t})$.
By Lemma \ref{lemma:indpred},
the right-hand side is equivalent to 
$[\Vec t] \in (\bigcup_k \varphi^k(\Vec\emptyset))_i$.
Since $M$ is standard, it is equivalent to 
$M \models P_i\Vec t$.
By definition of $f$, 
it is equivalent to
$f(\Code{P_i\Vec t})=0$.

From the assumption
$\forall f (I(f) \to f(\Code A)=0)$,
we have $f(\Code A)=0$ and hence $M \models A$.

The direction $(1) \Rightarrow (2)$:
Fix $f$ and
assume 
$I(f)$,
in order to show  $f(\Code A)=0$.
For closed terms $t, u$,
we define $t \sim u$ by $f(\Code{t=u}) = 0$.
Define $U = \{ t ~|~ t \text{ a closed term}\}/ \sim$.
We write $[t]$ for the equivalence class
for a closed term $t$ and $\sim$.
We define $\semb{F}( [t_1], \ldots, [t_n]) = [F(t_1,\ldots,t_n)]$ for function symbols $F$.
We define
$( [t_1], \ldots, [t_n]) \in \semb{P}$ by $f(\Code{ P(t_1,\ldots,t_n) })=0$ for predicate symbols $P$.
We define
$M = (U, \semb{ \ })$.
By Lemma \ref{lemma:indpred}, $M$ is standard.

We will show
$M \models B$ iff $f(\Code B) = 0$
by induction on $B$.

Case when $B$ is $\forall x C(x)$.
$M \models \forall x C(x) \Leftrightarrow$
(for every closed term $t$, $M \models C(t)$) $\Leftrightarrow$ (by IH) 
(for every closed term $t$, $f(\Code{ C(t)})=0$)
$\Leftrightarrow f(\Code{ \forall x C(x)}) = 0$.

Case $B \equiv P_i\Vec t$.
By definition of $M$, $M \models P_i\Vec t$ iff $f(\Code{P_i\Vec t})=0$.

The other cases are similar.

From the assumption, we have $M \models A$.
Hence $f(\Code A)=0$.
\end{proof}

\vspace{0.5\baselineskip}

Finally, we have the truth predicate for \FOLID\ with standard models.

\begin{Th}\label{th:folid}
Define the predicate $\Tr$ by
$\Tr(x) \equiv \forall f (I(f) \to f(x)=0).$
Then a closed formula $A$ of \FOLID\ is valid in standard models 
iff $\Tr(\Code A)$ is true.
\end{Th}

\begin{proof}
By Proposition \ref{prop:equiv-term},
$M \models A$ for any standard model $M$ of $(\Sigma,\Phi)$
iff
$M \models A$ for any standard term model $M$ of $(\Sigma_c,\Phi)$.
By Proposition \ref{prop:henkin},
it is equivalent to
$\forall f (I(f) \to f(\Code{A})=0)$,
namely, $\Tr(\Code A)$.
\end{proof}

\begin{Cor}\label{cor:folid}
The validity of \FOLID\ in standard models is a $\Pi^1_1$ relation.
\end{Cor}

\begin{proof}
Since $\Tr$ is a $\Pi^1_1$ formula,
$M \models A$ for any standard model $M$ is a $\Pi^1_1$ relation.
\end{proof}

\vspace{0.5\baselineskip}

There is another way to show Corollary \ref{cor:folid}.
To define the truth predicate for a closed formula $A$,
we will consider $\Eval A=T$ for every given structure $M=(U,\Eval\ )$.
By 
the downward Skolem-L\"{o}wenheim theorem,
it is sufficient to
consider only countable structure for $M$.
Hence the truth predicate for FOL (namely, \FOLID\ without inductive
predicates) can be defined by $\forall M(\hbox{$A$ holds in $M$})$
and it is a $\Pi^1_1$ relation.
To obtain the truth predicate for \FOLID,
we add the following clause for ``$A$ holds in $M$''
to the definition of the truth predicate:

$P(\Vec t)$ holds in $M$ iff $P^{(k)}(\Vec t)$ holds in $M$
for some $k$,
where $P^{(k)}$ is the $k$-times unfolding of the inductive predicate $P$.

Since the additional universal quantifier in the righthand side is first-order, 
the truth predicate for \FOLID\ is shown to be also a $\Pi^1_1$ relation.

In this paper,
in order to define the truth predicate for \FOLID,
we chose the idea of the truth predicate for $\omega$-language
given on Page 348 of \cite{Girard} instead.
Since \FOLID\ may not be a $\omega$-language,
in order to apply his idea to \FOLID,
we need the equivalence between the validity in any models
and the validity in any term models (Proposition \ref{prop:equiv-term}).
The way of the truth definition in this paper may give
a possibility of another approach to truth definitions
of higher-order languages that are not $\omega$-languages.

\section{Infinite-Descent Proof System $\LKIDomega$}

In this section,
we will provide the definition of the infinite-descent proof system $\LKIDomega$ \cite{Brotherston2011}.

The proof systems \LKIDo\ and \CLKIDo\ are presented for the logic
\FOLID\ \cite{Brotherston2011}.
These systems are obtained by adding inference rules for inductive predicate symbols to Gentzen's sequent calculus LK.
The system \LKIDo\ formalizes the proofs by infinite descent, and allows
possibly infinite proof trees.

A sequent is defined as 
$\Gamma \vdash \Delta$
where 
$\Gamma$ and $\Delta$ are finite sets of formulas.
We give the inference rules of LK in Figure \ref{fig:LK}  (the function $\mathit{FV}$ returns the set of free variables in the set of formulas or terms).
For the substitution $\theta=[x_1:=t_1,\ldots,x_m:=t_m]$ and the formula $A$, we write $A[\theta]$ for the formula obtained from $A$ by replacing all free occurrences of $x_1,\ldots,x_m$ in $A$ by $t_1,\ldots,t_m$, respectively.
We write $\Gamma[\theta]$ for the set $\{A[\theta] ~|~ A \in \Gamma \}$.
Note that the contraction rule is implicitly included.
A formula newly introduced into the conclusion of each logical inference rule is called the \emph{principal formula} of the rule.

According to the production rules of inductive predicates, we add two types of rules to define \LKIDo\ from LK: one that introduces an inductive predicate into the antecedent of a sequent (left introduction) and another that introduces an inductive predicate into the succedent of a sequent (right introduction).

\begin{figure*}[htbp]
\textbf{Structural rules:}
{\prooflineskip
\begin{align*}
&\infer[\text{(\textsc{Axiom})}(\Gamma \cap \Delta \neq \emptyset)]{\Gamma \vdash \Delta}{}
 \qquad 
 \infer[\text{(\textsc{Wk})}(\Gamma' \subseteq \Gamma \text{ and } \Delta' \subseteq \Delta)]{\Gamma \vdash \Delta}{\Gamma' \vdash \Delta'}
\\
&\infer[\text{(\textsc{Cut})}]{\Gamma \vdash \Delta}{
\Gamma \vdash F,\Delta \quad \Gamma,F \vdash \Delta}
 \qquad
\infer[\text{(\textsc{Subst})}]{\Gamma[\theta] \vdash \Delta[\theta]}{\Gamma \vdash \Delta}
 \end{align*}
}

{\prooflineskip
\textbf{Logical rules:}
\begin{align*}
&\infer[(\neg \mathrm{L})]{\Gamma, \neg F \vdash \Delta}{\Gamma \vdash F, \Delta}
\quad
 \infer[(\neg \mathrm{R})]{\Gamma \vdash \neg F, \Delta}{\Gamma, F \vdash \Delta}
\quad
 \infer[(\vee \mathrm{L})]{\Gamma, F \vee G \vdash \Delta}{\Gamma, F \vdash \Delta \quad \Gamma, G \vdash \Delta}
\quad
 \infer[(\vee \mathrm{R})]{\Gamma \vdash F \vee G, \Delta}{\Gamma \vdash F, G, \Delta}
\\
&  \infer[(\wedge \mathrm{L})]{\Gamma, F \wedge G \vdash  \Delta}{\Gamma, F, G \vdash \Delta}
\qquad
 \infer[(\wedge \mathrm{R})]{\Gamma \vdash F \wedge G, \Delta}{\Gamma \vdash F, \Delta \quad \Gamma \vdash G, \Delta}
\qquad
 \infer[(\to \mathrm{L})]{\Gamma, F \to G \vdash \Delta}{\Gamma \vdash F, \Delta \quad \Gamma, G \vdash \Delta}
\\
& \infer[(\to \mathrm{R})]{\Gamma \vdash F \to G, \Delta}{\Gamma, F \vdash G, \Delta}
\qquad
 \infer[(\forall \mathrm{L})]{\Gamma, \forall x F \vdash \Delta}{\Gamma, F[x:=t] \vdash \Delta}
\qquad
\infer[(\forall \mathrm{R})]{\Gamma \vdash \forall x F, \Delta}{\Gamma \vdash F, \Delta}
 \parbox{14ex}{($x \notin \mathit{FV}(\Gamma \cup \Delta)$)}
%\infer[(\forall \mathrm{R}) \text{ where } x \notin \mathit{FV}(\Gamma \cup \Delta)]{\Gamma \vdash \forall x F, \Delta}{\Gamma \vdash F, \Delta}
\\
& \infer[(\exists \mathrm{L}) (x \notin \mathit{FV}(\Gamma \cup \Delta))]{\Gamma, \exists x F \vdash  \Delta}{\Gamma, F \vdash \Delta}
\qquad
 \infer[(\exists \mathrm{R})]{\Gamma \vdash \exists x F, \Delta}{\Gamma \vdash F[x:=t], \Delta}
\\
& \infer[(= \mathrm{L})]{\Gamma[x:=t,y:=u], t=u \vdash \Delta[x:=t,y:=u]}{\Gamma[x:=u,y:=t] \vdash \Delta[x:=u,y:=t]}
\qquad
 \infer[(= \mathrm{R})]{\Gamma \vdash t=t, \Delta}{}
\end{align*}
}
\caption{Inference rules of LK}
\label{fig:LK}
\end{figure*}

Consider the production rule of the form (1) in Definition \ref{def:production_rules} as the $r$-th rule of the production rules that have $P_i$ as the conclusion.
The corresponding right introduction rule of $P_i$ is:
\vskip -1ex
{\small
\[
 \infer[(P_i R_r)]{
\Gamma \vdash P_i \Vec{t}(\Vec{u}), \Delta
}{
%  \deduce{
%  }{
  \Gamma \vdash Q_1 \Vec{u_1}(\Vec{u}), \Delta 
  \quad \ldots \quad 
  \Gamma \vdash Q_h \Vec{u_h}(\Vec{u}), \Delta
&
  \Gamma \vdash P_{j_1} \Vec{t_1}(\Vec{u}), \Delta
  \quad \ldots \quad 
  \Gamma \vdash P_{j_m} \Vec{t_m}(\Vec{u}), \Delta 
%  }
}
\]
}

For left introduction rules of inductive predicates, \LKIDo\ uses the following \emph{case-split rules}:
\[
\infer[(\textsc{Case } P_i)]{
\Gamma, P_i \Vec{u} \vdash \Delta
}{
\text{case distinctions}
}
\]
where the case distinctions are defined for every production rule of every $P_j$ that is mutually dependent with $P_i$ as follows:
For the production rule of the form (1) in Definition \ref{def:production_rules}, we define the corresponding \emph{case distinction} as follows:
\begin{align*}
 \Gamma, \bu = \bt(\by), Q_1 \Vec{u_1}(\by),\ldots,Q_h \Vec{u_h}(\by),
P_{j_1} \Vec{t_1}(\by), \ldots, P_{j_m} \Vec{t_m}(\by) \vdash \Delta
\end{align*}
where $\by$ is a sequence of distinct variables with the same length as $\bx = \mathit{FV}(\{\Vec t, \Vec{u_1},\ldots,\Vec{u_h},\Vec{t_1},\ldots,\Vec{t_m} \})$, and $z \notin \mathit{FV}(\Gamma \cup \Delta \cup \{P_i \Vec{u}\})$ for all $z \in \by$.
The formulas $P_{j_1} \Vec{t_1}(\by), \ldots, P_{j_m} \Vec{t_m}(\by)$ appearing in a case distinction are called the \emph{case-descendants} of the principal formula $P_i \bu$.

\begin{Eg}\rm
 The case-split rule corresponding to the inductive predicate $N$ from Example \ref{ex:nat} is as follows:
\[
 \infer[(\textsc{Case } N)]{
\Gamma, Nt \vdash \Delta
}{
\Gamma, t=0\vdash \Delta \quad \Gamma, t=sx, Nx \vdash \Delta
}
\]
The principal formula is $Nt$ and the case-descendant is $Nx$.
\end{Eg}

Next, we define proofs in \LKIDo. 
A (possibly infinite) derivation tree constructed according to the inference rules of \LKIDo\ is called a \emph{pre-proof of \LKIDo} if it does not have any open assumptions. An open assumption is a leaf node of the proof tree that is not an instance of an axiom.

A certain condition called the \emph{global trace condition} guarantees that a pre-proof is sound. For precisely describing it, we need to define traces.
 A (finite or infinite) \emph{path} in a derivation tree is a sequence of sequents $(S_i)_{0 \le i < \alpha} (\alpha \in \Nat \cup \{\infty\})$ such that for all $i+1<\alpha$, $S_{i+1}$ is a child of $S_i$ in the tree.
 An \emph{inductive atomic formula} is defined to be a formula of the form $P_i(\bt)$ where $P_i$ is an inductive predicate symbol.

 Let $\cD$ be a pre-proof of \LKIDo\ and $(\Gamma_i \vdash \Delta_i)_{i \ge 0}$ be a path in $\cD$.
A \emph{trace following} $(\Gamma_i \vdash \Delta_i)_{i \ge 0}$ is a sequence of inductive atomic formulas $(\tau_i)_{i \ge 0}$ such that  for all $i$, $\tau_i \in \Gamma_i$ and the following hold:
\begin{enumerate}

 \item If $\Gamma_i \vdash \Delta_i$ is the conclusion of rule (\textsc{Subst}), then $\tau_i = \tau_{i+1}[\theta]$, where $\theta$ is the substitution determined by this rule instance.

 \item If $\Gamma_i \vdash \Delta_i$ is the conclusion of rule (=L) and its principal formula is $s=t$, then $\tau_i = F[x:=t, y:=u]$ and $\tau_{i+1}=F[x:=u,y:=t]$ for some formula $F$ and variables $x,y$. 

 \item If $\Gamma_i \vdash \Delta_i$ is the conclusion of a case-split rule, then (a) $\tau_{i+1} = \tau_i$, or (b) $\tau_i$ is the principal formula of this rule instance and $\tau_{i+1}$ is the case-descendant of $\tau_i$. In the latter case, the occurrence $\tau_{i+1}$ is called a \emph{progress point} of the trace.

 \item If $\Gamma_i \vdash \Delta_i$ is the conclusion of the other rules, then $\tau_{i+1} = \tau_i$.
\end{enumerate}
A trace that has infinitely many progress points is called an \emph{infinitely progressing trace}.

 A pre-proof $\cD$ of \LKIDo\ is said to satisfy the \emph{global trace condition} if for every infinite path $(\Gamma_i \vdash \Delta_i)_{i \ge 0}$ in $\cD$, there exists an infinitely progressing trace that follows a tail path $(\Gamma_i \vdash \Delta_i)_{i \ge k}$ for some $k \ge 0$.
A pre-proof of \LKIDo\ is said to be a \emph{proof of \LKIDo} if it satisfies the global trace condition.

%\hrule

The system \CLKIDo\ is the proof system that only admits regular proof trees of \LKIDo, that is, trees having finitely many distinct subtrees.

\section{Logical Complexity of provability in $\LKIDomega$}

In this section, we will show that
the logical complexity of the provability in $\LKIDomega$ is $\Pi^1_1$-complete.

\subsection{The upperbound of $\LKIDomega$ provability}\relax
\label{sec:LKIDo_complexity}

In this section,
we will show $\LKIDomega$ provability is a $\Pi^1_1$ relation.

% From now on,
% since we will consider only closed formulas,
% for simplicity
% we write $\models$ for $\models_\rho$ with some variable assignment $\rho$.

The following fact is known for \LKIDo.

\begin{Th}[Theorem 5.9 in \cite{Brotherston2011}] \label{thm:completeness}\label{th:LKIDomega-complete}
$\Gamma \vdash \Delta$ is provable in \LKIDo\ iff
$\Gamma \vdash \Delta$ is valid in every standard model.
\end{Th}

Therefore, \LKIDo\ is a complete and sound proof system of \FOLID \ with respect to the standard interpretation of inductive predicates. 

\begin{Th} \label{corollary:Pi11}\label{cor:LKIDomega-Pi11}
The provability in $\LKIDomega$ is a $\Pi^1_1$ relation.
\end{Th}

\begin{proof}
Fix $(\Sigma,\Phi)$ and a sequent $\Gamma \vdash \Delta$.
Let
$\FV(\Gamma) \cup \FV(\Delta)=\{x_1,\ldots,x_n\}$.
Define
$A \equiv \forall x_1 \ldots \forall x_n (\bigwedge \Gamma \to \bigvee \Delta)$.
Then
$A$ is a closed formula of $\Sigma$
and
$\Gamma \prove \Delta$ is valid
iff
$A$ is valid.

From Corollary \ref{th:folid},
the validity of a given sequent in standard models is
a $\Pi^1_1$ relation.
By this and Theorem \ref{th:LKIDomega-complete},
we have the claim.
\end{proof}

\vspace{0.5\baselineskip}

From Theorem \ref{corollary:Pi11}, we can conclude that the
provabilities in the first-order logics with inductive definitions for
natural numbers, lists and trees are $\Pi^1_1$ relations.

\subsection{$\Pi^1_1$-hardness of $\LKIDomega$} \label{sec:LKIDo_hardness}

In this section,
we will show that the provability in 
$\LKIDomega$ is $\Pi^1_1$-hard.

We define $\PA+F$
as the logic obtained from {\bf PA} by adding a function symbol $F$,
which is an uninterpreted function symbol.
A {\em standard model} for Peano arithmetic {\bf PA}, Peano arithmetic $\PA+F$ with a new function symbol, and the second-order arithmetic
$\PA^2$ is defined as
a model such that
its universe is $\Nat$ and
$0$,$s$,$+$, and $\times$ are interpreted as
the ordinary constant and ordinary functions in natural numbers.

First, we will show the next lemma.

\begin{Lemma} \label{lemma:PA_Pi-1-1}
The set of formulas that are true
in every standard model of $\PA+F$
is $\Pi^1_1$-hard.
\end{Lemma}

\begin{proof}
Define $S_0 = \{ \Code A ~|~ 
A\text{ is a formula of } \hbox{$\PA+F$}
\text{ that is } \text{true } \text{in } \text{every } 
\text{standard}\text{ model} \\ \text{of }
\PA+F \}$.
It is sufficient to show that
for any first-order formula $A(F,n)$ of $\PA+F$
such that $\FV(A(F,x)) \subseteq \{x\}$,
if
$S_A=\{n ~|~ \hbox{$\forall f (A(f,n))$ is true in the standard
model of $\PA^2$}\}$ (that is, $S_A$ ranges over all $\Pi^1_1$-sets),
there is a reduction $\sigma$ such that
$x \in S_A$ iff $\sigma(x) \in S_0$,
where $A(f,n)$ is defined as 
the formula obtained from the formula $A(F,n)$
by replacing the function symbol $F$ by the function variable $f$.

We define the reduction $\sigma(m) = \Code{ A(F,m)}$.
Then the goal is shown by
$m \in S_A$ 
$\Leftrightarrow$
($\forall f (A(f,m))$ is true in
the standard model $M$ of $\mathbf{PA^2}$)
$\Leftrightarrow$
($\forall \varphi: \Nat \to \Nat. M[f:=\varphi] \models A(f,m)$ 
in the standard model $M$ of $\PA)$
$\Leftrightarrow$ 
($A(F,m)$ is true in every standard model of $\PA+F$)
$\Leftrightarrow \Code{ A(F,m)} \in S_0$
$\Leftrightarrow \sigma(m) \in S_0$,
where
$M[f:=\varphi]$ is the model such that
the function symbol $f$ is interpreted as $\varphi$
and the other symbols are interpreted in the same way as $M$.
\end{proof}

\begin{Prop} \label{thm:Pi11_hard}
The provability in $\LKIDomega$ is $\Pi^1_1$-hard.
\end{Prop}

\begin{proof}
By Theorem \ref{thm:completeness},
the provability in $\LKIDomega$ is equivalent to
the validity of \FOLID.
Hence, it is sufficient to show that
the validity of \FOLID\ in every standard model
is $\Pi^1_1$-hard.
We will show it.
Let $\Sigma_\PA^F$ be the signature consisting of
a constant $0$, unary function symbols $s, F$, 
binary function symbols $+$, $\times$, 
and a unary inductive predicate symbol $N$.
Let the production rules $\Phi_N$ be those in Example \ref{ex:nat}.
We write (PA1)--(PA6) for the following Peano axioms:
\begin{align*}
&\forall x(s x \ne 0), \qquad
 \forall xy(s x=s y \imp x=y), \qquad
 \forall x(x + 0 = x), \\
&\forall xy(x + s y = s(x+y)), \qquad
 \forall x(x \times 0 = 0), \qquad
 \forall xy(x \times s y = x \times y + x).
\end{align*}

For a formula $A$ of $\PA+F$,
we define a transformation $A^N$ as follows:
\begin{align*}
% (R({\Vec t}))^N = R({\Vec t}) \qquad
%	\hbox{($R$ is the symbol $=$ or a predicate symbol <= \textcolor{red}{Only =?})},\\
& (t_1 = t_2)^N = (t_1 = t_2), \qquad
 (\neg A)^N = \neg A^N, \qquad
 (A_1 \Box A_2)^N = A_1^N \Box A_2^N \quad (\Box \in \{\wedge, \vee, \to\}), \\
& (\forall x A)^N = \forall x (Nx \to A^N), \qquad
 (\exists x A)^N = \exists x (Nx \wedge A^N).
\end{align*}

For a formula $A$ of $\PA+F$,
if $FV(A) \subseteq \{x_1,\ldots,x_n\}$,
then the equivalence of the following can be shown
in a similar way to Lemma 3.12 in \cite{Brotherston2011}.
We give its proof in Appendix \ref{app:B} (Lemma \ref{lem:b3}).

(1)
$A$ is true in every standard model of $\PA+F$.

(2) The sequent 
$\mathrm{(PA1)}^N,\ldots,\mathrm{(PA6)}^N,Nx_1,\ldots,Nx_n, \forall x (Nx \to NFx) \vdash A^N$
is valid in every standard model for 
$(\Sigma_\PA^F,\Phi_N)$.

By this and Lemma \ref{lemma:PA_Pi-1-1},
the validity of \FOLID\ in every standard model is $\Pi^1_1$-hard.
\end{proof}

\subsection{$\Pi^1_1$-completeness of $\LKIDomega$}

In this section,
we will finally show that the provability in 
$\LKIDomega$ is $\Pi^1_1$-complete.

By Theorem \ref{cor:LKIDomega-Pi11} with Proposition \ref{thm:Pi11_hard},
we can show the complexity of provability of $\LKIDomega$.
\begin{Th}
The provability in $\LKIDomega$ is $\Pi^1_1$-complete.
\end{Th}

\section{Related work}
\label{sec:related}
The proof of the result of this paper was sketched before \cite{ito22},
and the present paper expands on the idea in detail,
specifically by providing a proof of
the equivalence between the validity in standard models and
the validity in standard term models,
as well as coding of inductive predicates.

We discuss related work other than \cite{ito22}.
Infinite-descent proof systems and
cyclic proof systems are widely studied.
For example, prior works have
compared their provabilities with traditional
inductive-definition systems in the style of Martin-L\"{o}f
\cite{Berardi2017a,Simpson2017,Berardi2017b,Das2020}, and explored their
applications to the logic for program termination
\cite{Brotherston2008} and separation logic \cite{Tatsuta2019,Ta2018}.
In particular,
infinite-descent proofs have been examined
in \cite{Brotherston2011,PhDBrotherston,Schopp02,Sprenger02,Sprenger03,Dam02,Niwinski96}.
However, 
these works focus exclusively on their relationship to cyclic proofs; neither 
the truth predicate associated with infinite-descent proofs nor their logical 
complexity has yet been investigated.

If the universe of the structures we consider consists of standard natural numbers,
the logical complexity of the system can be easily derived, since
we can represent an inductive predicate by some $\Pi^0_1$ formula
by taking the union of $k$-unfoldings of a given inductive predicate \cite{troelstrabook}.
However, this technique cannot be applied to $\LKIDomega$, since the universe may not consist of 
standard natural numbers.

It is well known that
the provability in Peano arithmetic with $\omega$-rules
is proved to be a $\Pi^1_1$ relation by chasing infinite recursive proof figures.
However, this technique cannot be applied to $\LKIDomega$ to show that
its provability is $\Pi^1_1$ for the following reason.
If we apply it, we need a second-order universal quantifier for
every infinite branch, and a second-order existential quantifier for a progressing trace,
so we can only show that the relation is $\Pi^1_2$.

The logical complexity of proof systems has been actively studied 
and the following are known.
The provability in $\LKID$ and $\CLKIDomega$ is 
$\Sigma^0_1$-complete \cite{ito22}.
The provability in first-order logic is $\Sigma^0_1$-complete \cite{DL}.
The provability in \textbf{PRA} is $\Sigma^0_1$-complete \cite{Girard}.
The provability in
second-order systems $\mathbf{PRA^2}$, $\mathbf{PA^2}$, and 
the validity of $\omega$-languages in $\omega$-models is $\Pi^1_1$-complete \cite{Girard}.
Apart from these, Dynamic Logic and the constructive infinitary logic $L_{\omega_1^\mathrm{ck}\omega}$ are known to have the provability that is $\Pi^1_1$-complete \cite{DL,Keisler,Harel85}.
The logical complexity of some infinite proof systems has also been studied. A fragment of linear logic extended by fixpoints, called $\mu \mathsf{MALL}^\infty$, was shown to have $\Pi^0_1$-hard provability \cite{DBLP:conf/fscd/00020S22}. This result was later sharpened to $\Sigma^1_1$-hardness with a $\Pi^1_2$ upper bound \cite{DBLP:conf/fsttcs/00020S23}.

Truth predicates have been used for analyzing logical complexity.
McGee \cite{McGee} proved that the validity of the modal 
predicate calculus is $\Pi^1_2$-complete 
by defining a truth predicate
for the logic in $\Pi^1_2$ formula.
Plisko \cite{PLISKO2001243} used the truth predicate in the analysis 
of logical complexity for complete constructive arithmetic theories.
Both works do not consider inductive predicates.

\section{Conclusion}
\label{sec:conclusion}

In this paper, we investigated the logical complexity of provability in \LKIDo,
a logical system of infinite-descent proofs for inductively defined
predicates, and proved that the complexity of provability in \LKIDo\
is $\Pi^1_1$-complete. To achieve this result, we defined a truth
predicate for the first-order language with inductive definitions using
a $\Pi^1_1$ formula.

There is also an alternative way to prove result (2) beyond what is
presented in this paper. To define the truth predicate, we consider
whether a formula holds in every given structure.
By 
the downward Skolem-L\"{o}wenheim theorem,
it is sufficient to
consider only countable structure.
Hence the truth predicate for FOL (namely, \FOLID\ without inductive
predicates) can be defined by 
``a formula holds in every countable model''
and it is a $\Pi^1_1$ relation.
To obtain the truth predicate for \FOLID,
we add the clause
``$P(\Vec t)$ holds iff $P^{(k)}(\Vec t)$ holds for some $k$''
to the definition of the truth predicate,
where $P^{(k)}$ is the $k$-times unfolding of the inductive predicate $P$.
Since the additional universal quantifier in the righthand side is first-order, 
the truth predicate for \FOLID\ is shown to be also a $\Pi^1_1$ relation.

Possible future research directions include investigating whether the
logical complexity of these systems changes when the signature is
restricted. Future work could also involve applying the proof technique
of this paper to other proof systems.

\ifdraft
\input{bibliography}
\else
\bibliographystyle{eptcs}
\bibliography{main}
\fi
\appendix
\section{Syntax and Semantics of \FOLID}\label{app:A}
We define the syntax and semantics of \FOLID.
Let $\Sigma$ be a signature consisting of variables $x_1,x_2,\ldots$,
function symbols $f_1,f_2,\ldots$,
and predicate symbols $R_1, R_2, \ldots$.

The terms of \FOLID\ are defined by the following grammar:
\[
 t ::= c_i \mid x_i \mid f_i(t_1,\ldots,t_n),
\]
where the arity of $f_i$ is $n$.

The formulas of \FOLID\ are defined by the following grammar:
\[
 A ::= R_i(t_1,\ldots,t_n) \mid \neg A \mid A \vee A \mid A \wedge A \mid
 A \to A \mid \exists x_i A \mid \forall x_i A, 
\]
where the arity of $R_i$ is $n$.

A structure $M$ of \FOLID\ is $(U, \Eval\ )$, where $U$ is a set called
the universe of $M$, and $\Eval\ $ is the interpretation of symbols
such that $\Eval{c_i} \in U$, $\Eval{f_i}: U^n \to U$ and 
$\Eval{R_i} \subseteq U^n$
($n$ is the arity of $f_i$ and $R_i$).
Let $\rho$ be a variable assignment, that is to say, a map from the set of 
variables to $U$.

Then the interpretation of terms $\Eval{t}\rho$ is defined as follows:
\begin{align*}
 \Eval{c_i}\rho &= \Eval{c_i}, \\
 \Eval{x_i}\rho &= \rho(x_i), \\
 \Eval{f_i(t_1,\ldots,t_n)}\rho &= \Eval{f_i}(\Eval{t_1}\rho,\ldots,\Eval{t_n}\rho).
\end{align*}
We may omit $\rho$ and simply write $\Eval{t}$ if $t$ is closed.

The relation $M \models_\rho A$ is inductively defined as follows:

\begin{tabular}{rcl}
 $M \models_\rho R_i(t_1,\ldots,t_n)$ & iff & $(\Eval{t_1}\rho, \ldots, \Eval{t_n}\rho) \in \Eval{R_i}$, \\
 $M \models_\rho \neg A$ & iff & $M \not \models_\rho A$, \\
 $M \models_\rho A \vee B$ & iff & $M \models_\rho A$ or $M \models_\rho B$, \\
 $M \models_\rho A \wedge B$ & iff & $M \models_\rho A$ and $M \models_\rho B$, \\
 $M \models_\rho A \to B$ & iff & $M \not \models_\rho A$ or $M \models_\rho B$, \\
 $M \models_\rho \exists x_i A$ & iff & $M \models_{\rho[x_i:=u]} A$ for some $u \in U$, \\
 $M \models_\rho \forall x_i A$ & iff & $M \models_{\rho[x_i:=u]} A$ for all $u \in U$,
\end{tabular}\\
where $\rho[x_i := u]$ is a variable assignment that maps $x_i$ to $u$ and $x_j (i \neq j)$ to $\rho(x_j)$.

\section{Proof of the claim in Proposition \ref{thm:Pi11_hard}}\label{app:B}
We give the proof of the claim in Proposition \ref{thm:Pi11_hard}.

We define $(F)$ as the formula $\forall x (Nx \to NFx)$.
Then, in any standard model $M = (U, \Eval{ \ })$ of $(\Sigma_\PA^F, \Phi_N)$ satisfying
$\mathrm{(PA1)}^N$--$\mathrm{(PA6)}^N$, it is clear that $\Eval{N} \cong \Nat$.
Therefore, any map $\rho$ from variables to $\Nat$ can be interpreted as
a variable assignment on such models $M$ as well as on standard models of $\PA + F$.

We will first show that the interpretations of terms in such models $M$ are
contained in the interpretation of the inductive predicate $N$.

\begin{lemma}\label{lem:b1}
 Let $M = (U, \Eval{ \ })$ be a standard model of $(\Sigma_\PA^F, \Phi_N)$ satisfying
$\mathrm{(PA1)}^N$--$\mathrm{(PA6)}^N$ and $(F)$.
Then, for any term $t$ and a variable assignment $\rho$, we have $\Eval{t}\rho \in \Eval{N}$.
\end{lemma}

\begin{proof}
By induction on $t$.

- $t \equiv 0$. Clearly, $\Eval{0}\rho \in \Eval{N}$ by the production rule.

- $t \equiv t_1 + t_2$. $\Eval{t_1 + t_2}\rho = \Eval{t_1}\rho + \Eval{t_2}\rho$.
By induction hypothesis, $\Eval{t_1}\rho, \Eval{t_2}\rho \in \Eval{N}$.
By Peano axioms and the fact that $\Eval{N}$ is the least fixpoint,
 $+$ is closed in $\Eval{N}$.
Hence, we have $\Eval{t_1}\rho + \Eval{t_2}\rho \in \Eval{N}$.

- $t \equiv t_1 \times t_2$. The same as the case $t \equiv t_1 + t_2$.

- $t \equiv Fu$. $\Eval{Fu}\rho = \Eval{F}(\Eval{u}\rho) \in \Eval{N}$ by induction hypothesis and $M \models (F)$.
\end{proof}

\vspace{0.5\baselineskip}

Next, we will show the equivalence between the validity of $\PA + F$ and
the truth in standard models of $(\Sigma_{\PA}^F, \Phi_N)$ with Peano
axioms and $(F)$.

\begin{lemma}\label{lem:b2}
 Let $\rho$ be a variable assignment on standard models of $\PA + F$.
For any formula $B$ of $\PA+F$, we have the following equivalence:

(1) For any standard models $M = (U,\Eval{ \ })$ of $(\Sigma_\PA^F, \Phi_N)$ satisfying
$\mathrm{(PA1)}^N$--$\mathrm{(PA6)}^N$ and $(F)$, $M \models_\rho B^N$.

(2) For any standard models $\cN = (\Nat, \Eval{ \ }_\cN)$ of $\PA+F$, $\cN \models_\rho B$.
\end{lemma}

\begin{proof}
 By induction on $B$.

- $B \equiv (t_1 = t_2)$.
For all $M$, $M \models_\rho (t_1 = t_2)^N \Leftrightarrow$
for all $M$, $M \models_\rho t_1 = t_2 \Leftrightarrow$
$\Eval{t_1}\rho = \Eval{t_2}\rho \Leftrightarrow$
(by Lemma \ref{lem:b1})$\Eval{t_1}_\cN\rho = \Eval{t_2}_\cN\rho$  for all $\cN$
$\Leftrightarrow$
$\cN \models_\rho t_1 = t_2$ for all $\cN$.

- $B \equiv \neg A$.
For all $M$, $M \models_\rho (\neg A)^N \Leftrightarrow$
for all $M$, $M \not \models_\rho A^N \Leftrightarrow$ (by IH)
for all $\cN$, $\cN \not \models_\rho A \Leftrightarrow$
for all $\cN$, $\cN \models_\rho \neg A$.

- $B \equiv A_1 \Box A_2, \Box \in \{\wedge, \vee, \to\}$.
The same as the case $B \equiv \neg A$.

- $B \equiv \forall x A$.
For all $M$, $M \models_\rho (\forall x A)^N \Leftrightarrow$
for all $M$, $M \models_\rho \forall x (Nx \to A^N) \Leftrightarrow$
for all $M$, $M \models_{\rho[x:=n]} Nx \to A^N$ for all $n \in U \Leftrightarrow$
for all $M$, $M \models_{\rho[x:=n]} Nx$ implies $M \models_{\rho[x:=n]} A^N$ for all $n \in U \Leftrightarrow$
for all $M$, $n \in \Eval{N}$ implies $M \models_{\rho[x:=n]} A^N$ for all $n \in U \Leftrightarrow$ (by $\Eval{N} \cong \Nat$)
for all $M$, $M \models_{\rho[x:=n]} A^N$ for all $n \in \Nat \Leftrightarrow$ (by IH)
for all $\cN$, $\cN \models_{\rho[x:=n]} A$ for all $n \in \Nat \Leftrightarrow$
for all $\cN$, $\cN \models_\rho \forall x A$.
\end{proof}

\vspace{0.5\baselineskip}

Now we will show our claim.

\begin{lemma}\label{lem:b3}
For a formula $A$ of $\PA+F$,
if $FV(A) \subseteq \{x_1,\ldots,x_n\}$,
then we have the following equivalence:

(1)
$A$ is true in every standard model of $\PA+F$.

(2) The sequent 
$\mathrm{(PA1)}^N,\ldots,\mathrm{(PA6)}^N,(F), Nx_1,\ldots,Nx_n \vdash A^N$
is valid in every standard model for 
$(\Sigma_\PA^F,\Phi_N)$.
\end{lemma}

\begin{proof}
 (1) is a restatement of the fact that
for every standard model $\cN$ of $\PA+F$ and all $\cN$-variable assignment $\rho$,
$\cN \models_\rho A$.
By Lemma \ref{lem:b2}, this is equivalent to the fact that
for every standard model $M = (U, \Eval{ \ })$ of $(\Sigma_\PA^F, \Phi_N)$ satisfying
$\mathrm{(PA1)}^N$--$\mathrm{(PA6)}^N$ and $(F)$,
and all $M$-variable assignment $\rho$ such that $\rho(x_i) \in \Nat (1 \le i \le n)$,
 $M \models_\rho A^N$.
Since $\Eval{N} \cong \Nat$, it is equivalent to the validity of the sequent
$Nx_1,\ldots,Nx_n \vdash A^N$.
As $\mathrm{(PA1)^N},\ldots,\mathrm{(PA6)^N}$ and $(F)$ are closed formulas,
this is equivalent to
$\mathrm{(PA1)^N},\ldots,\mathrm{(PA6)^N},(F), Nx_1,\ldots,Nx_n \vdash A^N$.
\end{proof}

\end{document}